\documentclass[conference]{IEEEtran}
\IEEEoverridecommandlockouts
\usepackage{cite}
\usepackage{amsthm,amsmath,amssymb,amsfonts}
\usepackage{algorithmic}
\usepackage{graphicx}
\usepackage{textcomp}
\usepackage{xcolor}
\usepackage{flushend}
\usepackage{enumerate}
\newtheorem{thm}{Theorem}
\newtheorem{pro}{Proposition}
\newtheorem{lem}{Lemma}

\theoremstyle{definition}
\newtheorem{defn}{Definition}

\newtheorem{algm}{Algorithm}

\theoremstyle{remark}
\newtheorem{remark}{Remark}
\def\BibTeX{{\rm B\kern-.05em{\sc i\kern-.025em b}\kern-.08em
    T\kern-.1667em\lower.7ex\hbox{E}\kern-.125emX}}
\begin{document}

\title{$\omega-$nonblocking supervisory control of discrete-event systems with infinite behavior\\
\thanks{This work was supported in part by the National Natural Science Foundation of China under Grant 62003199, in part by the Fundamental Research Funds for the
Central Universities of China under Grant 3102019ZDHKY11, in part by the Postdoctoral Science Foundation of China under Grant 2019M663819.}
}

\author{
\IEEEauthorblockN{Ting Jiao\IEEEauthorrefmark{1},  Renyuan Zhang\IEEEauthorrefmark{2}, Kai Cai\IEEEauthorrefmark{3}}

\IEEEauthorblockA{\IEEEauthorrefmark{1}School of Automation and Software Engineering, Shanxi University, Taiyuan, China}
\IEEEauthorblockA{\IEEEauthorrefmark{2}School of Automation, Northwestern Polytechnical University, Xi'an, China}
\IEEEauthorblockA{\IEEEauthorrefmark{3}Urban Research Plaza, Osaka City University, Japan}
}

\maketitle

\begin{abstract}
In the supervisory control framework of discrete-event systems (DES) with infinite behavior initiated by Thistle and Wonham, a supervisor satisfying the minimal acceptable specification and the maximal legal specification is synthesized. However, this supervisor may incur livelocks as it cannot ensure that the infinite behavior under supervision will always visit some marker states. To tackle this problem, we propose the definition of markability by requiring that all infinite cycles include at least one marker state. Then we formulate the problem of $\omega-$nonblocking supervisory control of DES with infinite behavior to synthesize an $\omega-$nonblocking (i.e. nonblocking, deadlock-free and livelock-free) supervisor. An algorithm is proposed to achieve $\omega-$nonblockingness by computing the supremal $*-$controllable, $*-$closed, $\omega-$controllable and markable sublanguage. We utilize the example of a robot as a running example.
\end{abstract}


\section{Introduction}
In the supervisory control framework of discrete event systems (DES), the role of a supervisor is to confine the behavior of the plant within a specified range prescribed by the specification. This is realized by disabling the occurrence of some controllable events. How to tackle the supervisory control of DES with finite behavior has been studied extensively. The definitions of $*-$controllability and $*-$closedness \cite{Wonham16a} are defined to achieve a marking nonblocking supervisory control, whose solvability is characterized in terms of the supremal $*-$controllable and $*-$closed sublanguage of the set of all allowable strings \cite{Wonham16a,Cassandras08}. For DES with infinite behavior, it is possible to investigate the supervisory control dealing with both safety and liveness specifications. The safety specification describes that some states should be avoided and the liveness specification describes that some states must be visited eventually (infinitely often). The solvability of the supervisory control problem for $\omega-$languages (SCP$^\omega$) is equivalent to the existence of an $\omega-$controllable and $\omega-$closed language. However, $\omega-$controllablility is preserved under arbitrary unions but not intersections, and $\omega-$closure is preserved under arbitrary intersections but not unions. Therefore, $\omega-$controllability and $\omega-$closure have to be considered separately \cite{Thistl91,ThiWon94a,ThiWon94b}. This incurs that the solvability of the SCP$^\omega$ is not a simple extension of the results obtained in DES with finite behavior.

In terms of the system behavior, the safety specification restrains the finite behavior and the liveness specification restrains the infinite behavior. For the finite behavior of the controlled DES, marker states and relevant marker behavior need to be considered, because marker states often embody some meaningful information, such as the completion of a task, distinguishing successful strings from others, etc. The definition of nonblocking requires that all strings surviving under control be prefixes of the marker behavior. In the existing work, the finite behavior of the controlled DES is nonblocking and meets the safety specification \cite{Wonham16a,Cassandras08,Thistl91,ThiWon94a,ThiWon94b,Seow20,Majumdar20,Ciolek20,Schmuck20}; the infinite behavior of the controlled DES is deadlock-free and meets the liveness specification \cite{Thistl91,ThiWon94a,ThiWon94b,Seow20,Majumdar20,Ciolek20,Schmuck20}. However, it is not guaranteed that the infinite behavior under supervision is able to visit given marker states, because by the nonblocking property only finite behavior can visit marker states, and meanwhile deadlock-freeness does not take marker states into consideration.

To meet the nonblocking requirement for DES with infinite behavior, we propose the definition of livelock-freeness for DES with infinite behavior, which describes that all infinite strings must visit at least one of the marker states. With this new definition, we propose the definition of $\omega-$nonblockingness for a supervisor if it is nonblocking, deadlock-free and livelock-free. Namely, an $\omega-$nonblocking supervisor will ensure that the finite behavior of the controlled DES is nonblocking; simultaneously, the infinite behavior of the controlled DES is both deadlock-free and livelock-free. Furthermore, we formulate the $\omega-$nonblocking supervisory control problem as: to construct an $\omega-$nonblocking supervisor such that the safety and liveness specifications are satisfied. To the best of the authors' knowledge, this problem has not been investigated in the literature.

In \cite{Wonham16a} and references therein, both the closed language and marker language are utilized to analyze the controlled behavior of the plant. The existence and synthesis results of state feedback for marker-progressive control of fair DES are presented in \cite{Seow20}, where the supervisory control problem is formulated in linear-time temporal logic. In \cite{Majumdar20}, an obliging game is employed to solve the supervisory control problem over non-terminating processes modeled as $\omega-$regular automata. The marker states in \cite{Majumdar20} are specified by a B\"uchi acceptance condition, thereby ensuring the infinite visit of marker states. A compositional supervisory control approach is proposed in \cite{Ciolek20} by utilizing the reactive synthesis and automated planning. Reference \cite{Schmuck20} demonstrates the relation between the reactive synthesis and supervisory control of non-terminating processes. Both \cite{Ciolek20} and \cite{Schmuck20} endeavor to analyze the supervisory control problem for DES with infinite behavior via automated planning and reactive synthesis methods and pay no special attention to marker states.

In this paper, we propose the definition of markability to describe livelock-freeness from the perspective of languages and an approach to compute the supremal markable sublanguage. Then, we propose an algorithm to compute the supremal language pair, i.e. the supremal $*-$controllable and $*-$closed sublanguage and the supremal markable and $\omega-$controllable sublanguage. We show that the $\omega-$nonblocking supervisory control problem of DES with infinite behavior is solvable if and only if the supremal language pair are both nonempty, and the infimal $\omega-$closed superlanguage of the minimal acceptable sublanguage belongs to the supremal markable and $\omega-$controllable sublanguage.

The main contributions of this paper are fourfold.
\begin{enumerate}
  \item We propose the definition of $\omega-$nonblockingness to ensure that the finite behavior of the controlled DES is nonblocking and the infinite behavior of the controlled DES is both deadlock-free and livelock-free.
  \item Based on the definition of $\omega-$nonblockingness, we formalize the $\omega-$nonblocking supervisory control problem of DES with infinite behavior. Other than satisfying the safety and liveness specifications, the $\omega-$nonblocking supervisory control requires that the infinite behavior under supervision visit some marker states.
  \item To depict livelock-freeness from the perspective of languages, we propose the definition of markability of a given $\omega-$language to ensure that the infinite strings will always visit some marker states. This definition differs from the marker states in the definition of $\mathcal{M}-$directingness \cite{{Seow20}}, where the marker states are similar to states in the B\"uchi acceptance criterion. Moreover, properties and computation approaches related to markability are also presented.
  \item For the $\omega-$nonblocking supervisory control of DES with infinite behavior, an approach to synthesize an $\omega-$nonblocking (i.e. nonblocking, deadlock-free and livelock-free) supervisor is proposed and an algorithm to compute it is presented. Our approach remedies the shortcoming that $\omega-$nonblockingness is not guaranteed in the infinite behavior obtained by the Thistle's approach \cite{Thistl91,ThiWon94a,ThiWon94b}. Moreover, our approach imposes no additional restrictions compared with the Seow's approach \cite{Seow20}.
\end{enumerate}

The rest of this paper is organized as follows. Section \uppercase\expandafter{\romannumeral2} lays a supervisory control background for DES with finite and infinite behavior. Section \uppercase\expandafter{\romannumeral3} formulates the $\omega-$nonblocking supervisory control of DES with infinite behavior. Section \uppercase\expandafter{\romannumeral4} presents the definition of markability and its relevant properties. Section \uppercase\expandafter{\romannumeral5} shows how to compute an $\omega-$nonblocking supervisor. Section \uppercase\expandafter{\romannumeral6} concludes this paper.

\section{Preliminaries}\label{preliminaries}
Let $\Sigma$ be a finite alphabet. Let $\Sigma^*$ and $\Sigma^\omega$ denote the sets of all finite and infinite strings over $\Sigma$ respectively. Let $\Sigma^\infty:=\Sigma^*\dot{\cup}\Sigma^\omega$.

For any $k\in\Sigma^*,v\in\Sigma^\infty$, write $k\leq v$ if there exists some $t\in\Sigma^\infty$ such that $kt=v$, i.e. $k$ is a \emph{prefix} of $v$. Define the map $\text{pre}:2^{\Sigma^\infty}\rightarrow2^{\Sigma^*}$ by
\begin{align*}
\text{pre}:V\mapsto\{k\in\Sigma^*|(\exists v\in V)k\leq v\}.
\end{align*}
The \emph{limit} of a $*-$language is given by
\begin{align*}
\text{lim}(K):=\text{pre}^{-1}(K)\cap\Sigma^\omega,
\end{align*}
where $\text{pre}^{-1}:2^{\Sigma^*}\rightarrow2^{\Sigma^\infty}$ is the inverse image of $\text{pre}:2^{\Sigma^\infty}\rightarrow2^{\Sigma^*}$.

Define operator $\text{clo}:2^{\Sigma^\infty}\rightarrow2^{\Sigma^\infty}$ as
\begin{align*}
\text{clo}:R\mapsto\text{lim}(\text{pre}(R))=\text{pre}^{-1}(\text{pre}(R))\cap\Sigma^\omega.
\end{align*}
$\text{clo}(R)$ is called the $\omega-$\emph{closure} of $R$. $R$ is $\omega-$\emph{closed} if $R=\text{clo}(R)$. $R$ is $\omega-$\emph{closed with respect to} $S$ if $R=\text{clo}(R)\cap S$, where $S\subseteq\Sigma^\omega$.

The plant to be controlled is modeled by a six-tuple \[{\bf G} = (Q,\Sigma,\eta,q_0,Q_m,\mathcal{B}_G),\]
where $Q$ is the \emph{finite state set}, $q_0$ is the \emph{initial state}, $\Sigma$ is the \emph{finite event set}, $\eta:Q\times\Sigma\rightarrow Q$ is the (partial) \emph{state transition function}, $Q_m$ is the \emph{set of marker states}, and $\mathcal{B}_G$ is the \emph{B\"uchi acceptance criterion}. For plant $\bf G$, the finite behavior is modeled by the five-tuple $(Q,\Sigma,\eta,q_0,Q_m)$; its finite closed behavior and finite marker behavior are denoted as $L({\bf G}):=\{s\in\Sigma^*|\eta(q_0,s)!\}$ and $L_m({\bf G}):=\{s\in L({\bf G})|\eta(q_0,s)\in Q_m\}$ respectively, where $\eta(q_0,s)!$ means that $\eta(q_0,s)$ is defined. The infinite behavior of $\bf G$ is modeled by the five-tuple $(Q,\Sigma,\eta,q_0,\mathcal{B}_G)$ and is denoted as $S({\bf G}):=\{s\in\Sigma^\infty|R(s)\cap\mathcal{B}_G\neq\emptyset\}$, where $R(s)$ is the set of states that string $s$ visits infinitely often. In this paper, we only consider the case that $\mathcal{B}_G=Q$. Namely, the infinite behavior can be interpreted as an absence of liveness assumptions in the modeling of the uncontrolled DES. Thus, we have $S({\bf G})=\text{lim}(L({\bf G}))$. We say that ${\bf G}$ is {\it nonblocking} if $L({\bf G}) = \overline{L_m({\bf G})}$, and is {\it deadlock-free} if $\text{pre}(S({\bf G})) = L({\bf G})$.

In this paper, we only consider the infinite behavior depicted by deterministic B\"uchi automata (DBA). For a given DBA ${\bf G}$ over an alphabet $\Sigma$ and for any infinite string $s\in S({\bf G})$, infinite string $s\in\Sigma^\omega$ can be written as $s:=tu^\omega$ with $t,u\in \Sigma^*,u\neq\varepsilon$ as DBA are a special case of the nondeterministic B\"uchi automata (NBA), which agree with the class of $\omega-$regular languages \cite{Bai08}.

There are generally two classes of control requirements imposed on $\bf G$: {\it safety}
specifications describing that some conditions on $\bf G$ {\it must not} occur,
and {\it liveness} specifications describing that some other conditions
must occur {\it eventually} \cite{Lampor77}.

For safety specification $E_s\subseteq\Sigma^*$, a nonblocking supervisor
\begin{align*} \label{eq:f*}
f^*:L({\bf G}) \rightarrow \Gamma
\end{align*}
may be constructed (\cite{Wonham16a,GolRam87}) such that the finite closed-loop
marker behavior $L_m({\bf G}^{f^*})$ satisfies
\begin{align*}
 &L_m({\bf G}^{f^*}) = \sup\mathcal{C}^*(E_s) \subseteq E_s,
\end{align*}
where ${\bf G}^{f^*}$ represents the action of supervisor $f^*$, and $\sup\mathcal{C}^*(E_s)$ denotes the supremal $*-$controllable and $*-$closed sublanguage of $E_s$.

For \emph{maximal legal specification} $E_l\subseteq\Sigma^\omega$ and \emph{minimal acceptable specification} $A_l\subseteq\Sigma^\omega$, it is proved in \cite[Theorem 5.3]{ThiWon94b} that there exists an $\omega-$controllable and $\omega-$closed language $T$ such that $A_l \subseteq T \subseteq E_l$ if and only if
\begin{align*} 
\inf\mathcal{F}^\omega(A_l) \subseteq \sup\mathcal{C}^\omega(E_l),
\end{align*}
where $\inf\mathcal{F}^\omega(A_l)$ represents the infimal $\omega-$closed superlanguage of $A_l$ and $\sup\mathcal{C}^\omega(E_l)$ represents the supremal $\omega-$controllable sublanguage of $E_l$.

If such $T$ exists, a  deadlock-free supervisor
\begin{align*}
f^\omega : L({\bf G}) \rightarrow \Gamma
\end{align*}
may be constructed such that the infinite closed-loop behavior $S({\bf G}^{f^\omega})$ satisfies $S({\bf G}^{f^\omega}) = T$ and therefore
\begin{align*}
A_l \subseteq S({\bf G}^{f^\omega}) \subseteq E_l,
\end{align*}
where ${\bf G}^{f^\omega}$ represents the action of supervisor $f^\omega$. The detailed construction rules of the supervisor $f^\omega$ is referred to \cite{ThiWon94b}.

Because the overall supervisor should satisfy safety specifications and liveness specifications simultaneously, we define the {\it supervisory control} for $\bf G$ as any map $f: L({\bf G}) \rightarrow \Gamma$ with $f=f^*=f^\omega$, where $\Gamma := \{\gamma \subseteq \Sigma| \gamma \supseteq \Sigma_{u}\}$.
Then the finite and infinite closed-loop behaviors of the {\it controlled DES}
${\bf G}^f$, representing the action of the supervisor $f$ on $\bf G$,
are respectively given by

\begin{enumerate}[(a)]
\item $L({\bf G}^f)$, the finite closed behavior synthesized by $f$, defined
by the following recursion:
\begin{align*}
\mbox{(i)} ~&\varepsilon \in L({\bf G}^f), \notag\\
\mbox{(ii)} ~&(\forall s \in \Sigma^*, \forall \sigma\in \Sigma) ~s\sigma \in L({\bf G}^f)
\Leftrightarrow\\
\mbox{ } ~&s \in L({\bf G}^f) ~\&~ s\sigma \in L({\bf G}) ~\&~ \sigma \in f(s), \notag\\
\mbox{(iii)} ~&\mbox{no other strings belong to $L({\bf G}^f)$;}
\end{align*}

\item $L_m({\bf G}^f)$, the finite marker behavior synthesized by $f$, given by
\begin{align*}
L_m({\bf G}^f):= L({\bf G}^f)\cap L_m({\bf G});
\end{align*}

\item $S({\bf G}^f)$, the infinite behavior synthesized by $f$, given by
\begin{align*}
S({\bf G}^f):= \text{lim}(L({\bf G}^f)) \cap S({\bf G}).
\end{align*}
\end{enumerate}
In other words, $L({\bf G}^f)$ is the finite closed behavior under the control of supervisor $f$. Namely, the empty string $\varepsilon$ is in $L({\bf G}^f)$; event $\sigma\in\Sigma$ is enabled after the occurrence of string $s\in L({\bf G}^f)$ if and only if event $\sigma$ is defined after the occurrence of string $s\in L({\bf G})$ and event $\sigma$ is enabled by supervisor $f$. The finite marker behavior $L_m({\bf G}^f)$ consists exactly of the strings of $L_m({\bf G})$ that `survive' under supervision by $f$. The controlled infinite behavior $S({\bf G}^f)$ consists of the $\omega-$strings generated by finite closed behavior $L({\bf G}^f)$ and within the infinite behavior $S({\bf G}^f)$.
\section{Problem formulation}\label{problem_formulation}
Write \[{\bf G}^f = (Q^f,\Sigma,\eta^f,q_0^f,Q_m^f).\]
We say that $f: L(\bf G) \rightarrow \Gamma$ is a {\it nonblocking} supervisor if ${\bf G}^f$ is nonblocking (i.e. $\overline{L_m({\bf G}^f)} = L({\bf G}^f)$, and a {\it deadlock-free} supervisor if ${\bf G}^f$ is deadlock-free (i.e. $\text{pre}(S({\bf G}^f)) = L({\bf G}^f)$). However, the requirement that a supervisor $f$ be both nonblocking and deadlock-free is not enough for certain applications.
\subsection{Motivating example}
To motivate our work, we consider a robot throughout this paper, whose behavior is depicted by the automaton shown in Fig.~\ref{fig:robot}. The robot traverses five zones represented by states $0,1,2,3$ and $4$ respectively, and it has to return to zone $0$ regularly for recharging (denoted by marker state $0$). In addition, the robot may visit zone $3$ for temporary recharging (denoted by marker state $3$). Events $c_i,u_j,i\in\{1,\cdots,5\},j\in\{1,2,3\}$ mean that the robot leaves one room and enters into the other room. Events $c_i,i\in\{1,\cdots,5\}$ are controllable and events $u_j,j\in\{1,2,3\}$ are uncontrollable.

\begin{figure}[ht]
\centering
    \includegraphics[scale = 0.5]{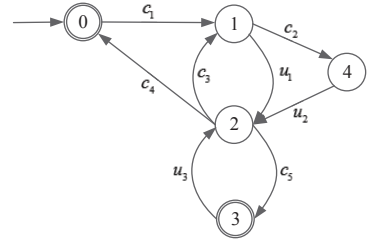}
\caption{State transition graph of a robot}
\label{fig:robot}
\end{figure}

The safety specification $E_s$ prohibits the robot from entering zone $4$, where a dangerous area is located. The safety specification $E_s$ is represented by the automaton shown in Fig.~\ref{fig:robot_es}. The maximal legal specification $E_l$ requires that the robot must always inspect zone $1$ to supervise all the activities in this critical zone. Accordingly, the maximal legal specification $E_l$ is represented by the DBA shown in Fig.~\ref{fig:robot_el}, whose B\"uchi acceptance criterion is $\{1\}$. By computing the finite closed-loop marker behavior (by the standard supervisory control theory of DES with finite behavior in \cite{Wonham16a}) and infinite closed-loop behavior (by the Thistle and Wonham's supervisory control theory of DES with infinite behavior in \cite{ThiWon94b}), we obtain the controlled DES ${\bf G}^f$ shown in Fig.~\ref{fig:robot_gf}.

\begin{figure}[ht]
\centering
    \includegraphics[scale = 0.5]{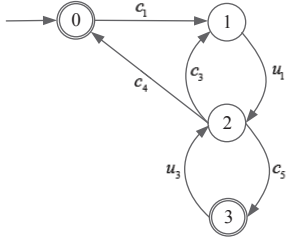}
\caption{State transition graph of the automaton representing safety specification $E_s$}
\label{fig:robot_es}
\end{figure}

\begin{figure}[ht]
\centering
    \includegraphics[scale = 0.5]{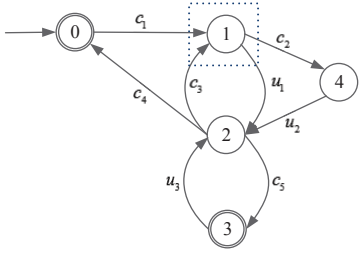}
\caption{State transition graph of the DBA representing maximal legal specification $E_l$}
\label{fig:robot_el}
\end{figure}

\begin{figure}[ht]
\centering
    \includegraphics[scale = 0.5]{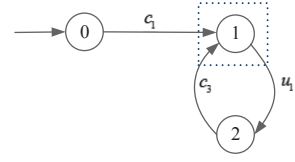}
\caption{State transition graph of the controlled DES ${\bf G}^f$}
\label{fig:robot_gf}
\end{figure}

It is inspected from Fig.~\ref{fig:robot_gf} that ${\bf G}^f$ is deadlock-free but blocking. Moreover, the robot has no chance to recharge by following the infinite behavior $c_1(u_1c_3)^\omega$. This implies that we cannot only resort to the Thistle's approach to simultaneously guarantee that the finite behavior of the controlled DES is nonblocking and any infinite string of the controlled DES with infinite behavior will visit some marker states. The undesirable infinite behavior is inevitable as marker states are ignored in the Thistle's approach.

Intuitively, we may add marker states into the B\"uchi acceptance criterion to have the marker states be visited infinitely often. However, by doing so, we will unnecessarily enlarge the maximal legal specification. Consequently, there may exist infinite strings of the controlled DES such that some marker states, rather than the states in the original B\"uchi acceptance criterion, are visited infinitely often. In the example of the robot, by incorporating the marker states into the B\"uchi acceptance criterion, the resultant supervisor synthesized by the Thistle's approach is shown in Fig. \ref{fig:robot_sup'}. From this figure we inspect that there exist infinite strings allowed by the supervisor, but failing to ensure that the states in the original B\"uchi acceptance criterion being visited infinitely often, say infinite string $c_1u_1c_5(u_3c_5)^\omega$. Alternatively, if we compute the intersection of the set of marker states and the B\"uchi acceptance criterion, then we may obtain an empty set, just as in the example of the robot, and no effective supervisor will be synthesized with the Thistle's approach.

\begin{figure}[ht]
\centering
    \includegraphics[scale = 0.5]{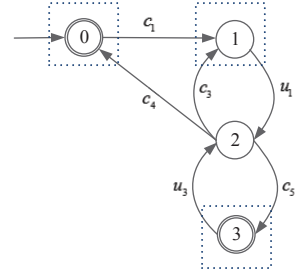}
\caption{State transition graph of the controlled DES when the marker states are added into the B\"uchi acceptance criterion in the example of the robot}
\label{fig:robot_sup'}
\end{figure}

Meanwhile, to perform the event-selection mechanism, the Seow's approach in \cite{Seow20} needs to choose the set of fair events $\Sigma_{\mathcal{F}}\subseteq\Sigma_u$. However, in the example of the robot, as events $c_4$ and $c_5$ are controllable, they are not potential fair events. Thus, to apply the Seow's approach, we have to modify controllable events $c_4$ and $c_5$ to uncontrollable events (say $u_4$ and $u_5$ respectively). By the detailed analysis explained in Remark \ref{rem_seow}, we have that the infinite behavior of the resultant supervisor is more restrictive than our approach. Therefore, how to guarantee that any infinite string in the controlled DES with infinite behavior will always visit some marker states motivates the work of this paper.

\subsection{Problem formulation}
By the above analysis, we know that in some applications, in addition to nonblocking and deadlock-free, it is required that the infinite strings of controlled plant ${\bf G}^f$ should always visit some marker states. This is the property of \emph{livelock-free}; namely, all the cycles representing infinite behavior of ${\bf G}^f$ must include at least one marker state. The general definition of \emph{livelock} is that a specific process is not progressing \cite{Sam80}. In DES, a livelock occurs if there exists a \emph{strongly connected component} such that there is no marker state in it and there is no transition defined exiting from it \cite{Yin16}. In our scenario, a livelock occurs if an infinite string cannot visit some marker states. 


\begin{defn}
Given controlled plant ${\bf G}^f$ defined as above, we say that ${\bf G}^f$ is
{\it livelock-free} if:
\[ (\forall s \in S({\bf G}^f))~ R(s) \cap Q_m^f \neq \emptyset,\]
where $R(s)$ represents the set of states visited by $s$ infinitely often
and $Q_m^f$ is the set of marker states of ${\bf G}^f$.
\end{defn}

Note here that this definition is different with its finite counterpart.
For finite behavior, ${\bf G}^f$ is livelock-free if there does not exist any cycle $C \subseteq Q^f$
such that (i) $C \cap Q_m^f = \emptyset$ and (ii) there does not
exist any transition exiting from $C$ \cite{Yin16}. Considering infinite behavior, the second condition is not required
because the system may visit the states in a cycle infinitely often and will never exit from the cycle.

\begin{defn}
Given controlled plant ${\bf G}^f$ defined as above, we say that ${\bf G}^f$ is
{\it $\omega-$nonblocking} if:
\begin{enumerate}[(a)]
  \item ${\bf G}^f$ is nonblocking, i.e. $L({\bf G}^f)=\overline{L_m({\bf G}^f)}$;
  \item ${\bf G}^f$ is deadlock-free, i.e. $\text{pre}(S({\bf G}^f))=L({\bf G}^f)$;
  \item ${\bf G}^f$ is livelock-free, i.e. $(\forall s\in S({\bf G}^f))R(s)\cap Q_m^f\neq\emptyset$.
\end{enumerate}
\end{defn}
From the definition of $\omega-$nonblocking, we know that for the finite behavior, ${\bf G}^f$ is nonblocking; for the infinite behavior, it is both deadlock-free and livelock-free.

To this end, we are ready to formulate the $\omega-$nonblocking supervisory control problem of DES with infinite behavior (NSCP$^\omega$) as follows:

{\it
Given a DES ${\bf G}$ defined above, $*-$language $E_s\subseteq \Sigma^*$ (representing safety specification),
and $\omega-$languages $A_l, E_l \subseteq \Sigma^\omega$ (representing minimal acceptable and maximal legal liveness specifications respectively) such that $E_s\subseteq L_m({\bf G})$ and $A_l\subseteq E_l\subseteq S({\bf G})$,
construct an $\omega-$nonblocking supervisor $f:L({\bf G})\rightarrow \Gamma$ for ${\bf G}$ such that
\begin{align}
& L_m({\bf G}^f) \subseteq E_s; \label{eq:NSCP_cond1}\\
& A_l \subseteq S({\bf G}^f) \subseteq E_l. \label{eq:NSCP_cond2}
\end{align}
}


\section{Definition and properties of markability}\label{markability}
To satisfy the livelock-free property for ${\bf G}^f$, we introduce the following definition of markability of a given $\omega-$language.

\begin{defn}\label{defn:markable}
For any DBA ${\bf G}$ and any $T\subseteq S({\bf G})$, $T$ is {\it markable} with respect to $L_m({\bf G})$ if
for any string $s \in T$, written as $s = tu^\omega$ ($t,u\in \Sigma^*,u\neq\varepsilon$),
there exists $s' \in \text{pre}(s)$ such that $s' \geq t$ and $s' \in L_m({\bf G})$.
\end{defn}

Definition \ref{defn:markable} can be interpreted as follows. For any infinite string $s=tu^\omega(t,u\in\Sigma^*,u\neq\varepsilon)$, at least one marker state will be reached by string $s' \in \text{pre}(s)$ with $s' \geq t$ and $s' \in L_m({\bf G})$. Afterwards, this marker state will be visited infinitely often as it is in the cycle of an infinite string. Such a cycle must exist as we only consider $\omega-$languages recognized by finite DBA. For the DBA shown in Fig.~\ref{fig:robot_el}, let $s=c_1(u_1c_4c_1)^\omega$ with $t:=c_1,u:=u_1c_4c_1$. There exists string $s^\prime=c_1u_1c_4\in\text{pre}(s)$ such that $s^\prime\geq t$ and $s^\prime\in L_m({\bf G})$. However, infinite string $s=c_1(u_1c_3)^\omega$ with $t:=c_1,u:=u_1c_3$ fails to satisfy the condition that there exists $s^\prime\in\text{pre}(s)$ such that $s^\prime\geq t$ and $s^\prime\in L_m({\bf G})$. This is due to the fact that no marker state will be visited infinitely often by infinite string $s=c_1(u_1c_3)^\omega$.

We note that by the definition of $\mathcal{M}-$directingness \cite{Seow20}, every legal state trajectory of DES model ${\bf G}$ satisfying $\Box P$ (resembling the requirement of the safety specification) also satisfies that every marker condition in system marker set $\mathcal{M}$ can be met infinitely often (resembling the requirement of the liveness specification). Thus, in the definition of $\mathcal{M}-$directingness, the marker states are in fact similar to states in the B\"uchi acceptance criterion. While in our definition of markability, we require that all infinite cycles include at least one marker states. Thus, markability imposes additional constraints on infinite behaviors. In our language setting, if we let the set of marker states coincide with the B\"uchi acceptance criterion, then all marker states will be visited infinitely often, thereby exhibiting the characteristic of $\mathcal{M}-$directingness.

The markable languages have the following properties.

\begin{lem} \label{lem:mark_sub}
If an $\omega-$language $S \subseteq \Sigma^\omega$ is markable with respect to $L_m({\bf G})$, then any sublanguage
$T \subseteq S$ is markable with respect to $L_m({\bf G})$.
\end{lem}

\begin{proof}
Since $T\subseteq S$, any infinite string $s \in T$ also
belongs to $S$, i.e. $s \in S$. Then the result is immediate because $S$ is markable.
\end{proof}

\begin{lem} \label{lem:mark_nonblock}
If an $\omega-$language $S \subseteq \Sigma^\omega$ is markable with respect to $L_m({\bf G})$, then
\[\overline{\text{pre}(S) \cap L_m({\bf G})} = \text{pre}(S).\]
\end{lem}

\begin{proof}
$(\subseteq)$ This direction is obvious because $\overline{\text{pre}(S)} = \text{pre}(S)$.\\
$(\supseteq)$ Let $s \in \text{pre}(S)$. Because $S$ is $\omega-$regular,
there must exist two strings $t,u\in \Sigma^*$ such that $stu^\omega \in S$. Since
$S$ is markable with respect to $L_m({\bf G})$, there must exist string $s' \in \text{pre}(stu^\omega)$ such that
$s' \geq st$ and $s' \in L_m({\bf G})$. By $s' \geq st$, we have $s \in \overline{s'}$. With $s' \in \text{pre}(stu^\omega) \subseteq \text{pre}(S)$ and $s' \in L_m({\bf G})$, we have $s\in \overline{\text{pre}(S) \cap L_m({\bf G})}$.
\end{proof}

Let
\begin{align*}
\mathcal{M}(E_l) = \{&T \subseteq \Sigma^\omega| T\subseteq  E_l \cap S({\bf G}) \\&\text{is markable with respect to $L_m({\bf G})$}\}
\end{align*}
be the set of markable sublanguages of the maximal legal specification $E_l$.

\begin{lem}\label{lem:sup_markability}
$\mathcal{M}(E_l)$ is nonempty and is closed under arbitrary unions. In particular, $\mathcal{M}(E_l)$ contains a (unique) supremal element, which we denote by $\text{sup}\mathcal{M}(E_l)$.
\end{lem}
\begin{proof}
Since the empty language is markable with respect to $L_m({\bf G})$, it is a member of $\mathcal{M}(E_l)$. Let $T_\alpha\in \mathcal{M}(E_l)$ for all $\alpha$ in some index set $A$, and let $T=\cup\{T_\alpha|\alpha\in A\}$. Then $T\subseteq E_l$. Furthermore, for any $s\in T$ written as $s=tu^\omega(t,u\in\Sigma^*)$, there exists some $\alpha\in A$ such that $s\in T_\alpha$. Thus, there exists $s^\prime\in\text{pre}(s)$ such that $s^\prime\geq t$ and $s^\prime\in L_m({\bf G})$ as $T_\alpha\in \mathcal{M}(E_l)$ is a markable sublanguage of $E_l$.\\
Finally we have for the supremal element
\begin{align*}
\text{sup}\mathcal{M}(E_l)=\cup\{T|T\in \mathcal{M}(E_l)\}.
\end{align*}
\end{proof}

Then $\text{sup}\mathcal{M}(E_l)$ (the supremal markable
sublanguage of $E_l$) can be computed as follows.

\begin{pro} \label{pro:mark}
For given DES ${\bf G} = (Q,\Sigma,\eta,q_0,Q_m,\mathcal{B}_{\bf G})$ and maximal legal specification $E_l$,
construct DBA ${\bf G}' = (Q,\Sigma,\delta,q_0,\mathcal{B}_{{\bf G}'})$ with $\mathcal{B}_{{\bf G}'}=Q_m$; then
\begin{align*}
\sup\mathcal{M}(E_l) = S({\bf G}')\cap E_l.
\end{align*}
\end{pro}

\begin{proof}
We first show that $E' := S({\bf G}')\cap E_l \in \mathcal{M}(E_l)$.

Let $s \in E'$. Since $E'$ is $\omega-$regular, there must exist $t,u\in \Sigma^*$ such that $s = tu^\omega$.
Since $s \in S({\bf G}')$, there must exist a state $q_s \in \mathcal{B}_{{\bf G}'}=Q_m$, which
is visited by $s$ infinitely often, i.e. $q_s\in \inf(s) = \{q\in \mathcal{B}_{{\bf G}'}|\exists^\omega n \in \mathbb{N}:s(n)=q\}$,
where $\exists^\omega$ denotes the quantifier ``there exist infinitely many", $\mathbb{N}$
denotes the set of natural numbers, and $s(n)$ represents the state visited by $s$ at the $n$-th step.
Since $tu^\omega \in S({\bf G}')$, $t \in \text{pre}(S({\bf G}'))$ and thus there must exist
a state $q_t \in Q$ satisfying $\delta(q_0,t) = q_t$. We have already known that
$tu^\omega$ will visit state $q_s$ infinitely often. So there must exist
$t' \in \text{pre}(u^\omega) \subseteq \Sigma^*$ such that $\delta(q_t,t') = q_s \in Q_m$, i.e. $\delta(q_0,tt')\in Q_m$;
thus $tt' \in L_m({\bf G})$. Obviously, $t \leq tt'$. Hence $E'$ is markable.

Now, suppose $E'' \in \mathcal{M}(E_l)$, and $s \in E''$;
we must show that $s \in E'$. Write $s = tu^\omega$. Since $E''$ is markable, there exists
$s' \in \text{pre}(s)$ such that $s' \geq t$ and $s' \in L_m({\bf G})$. So there
exist states $q_t$ and $q_{s'}$ such that $\delta(q_0,t) = q_t$ and $\delta(q_0,s')=q_{s'} \in Q_m$.
Because $s = tu^\omega \in E_l \cap S({\bf G})$, state $q_t$ and its downstream states $q'$ visited by
strings $t' \in \text{pre}(u^\omega)$ will be visited by $s$ infinitely many times, i.e. $q_t,q' \in \inf(s)$.
Furthermore, since $s'\geq t$, there must exist a string $t' \in \Sigma^*$ such that $s' = tt'$;
thus $\delta(q_t,t')=q_{s'}$. Obviously, $t' \in \text{pre}(u^\omega)$; thus $q_{s'}\in \inf(s)$.
Hence, $q_{s'}\in \inf(s) \cap Q_m = \mathcal{B}_{{\bf G}'}$. Namely, $\inf(s)\cap \mathcal{B}_{{\bf G}'} \neq \emptyset$;
thus $s$ is accepted by ${\bf G}'$, i.e. $s \in S({\bf G}')$, which derives that $s \in E'$ (as required).
\end{proof}

For an arbitrary maximal legal specification $E_l$, there may exist infinite behavior with no marker states in ${\bf G}$ being visited. By the intersection of $S({\bf G}')$ and $E_l$ as given in Proposition \ref{pro:mark}, the resultant supremal markable sublanguage of $E_l$ will ensure that the obtained infinite strings will always visit some marker states in ${\bf G}$.

Let the example of the robot be revisited. To compute the supremal element $\text{sup}\mathcal{M}(E_l)$, we first construct DBA ${\bf G}'$ (shown in Fig.~\ref{fig:robot_G'}) by the approach given in Proposition \ref{pro:mark}. Then we compute the intersection of $S({\bf G}')\cap E_l$. The result is represented by the DBA shown in Fig.~\ref{fig:robot_mel}, where infinite string $s=c_1(u_1c_3)^\omega$ no longer exists as only state $1$ will be visited infinitely often.
\begin{figure}[ht]
\centering
    \includegraphics[scale = 0.5]{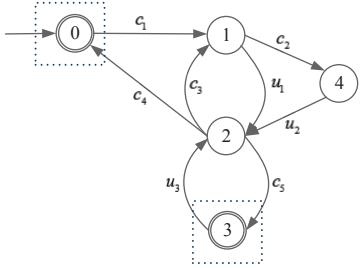}
\caption{State transition graph of DBA ${\bf G}'$}
\label{fig:robot_G'}
\end{figure}
\begin{figure}[ht]
\centering
    \includegraphics[scale = 0.5]{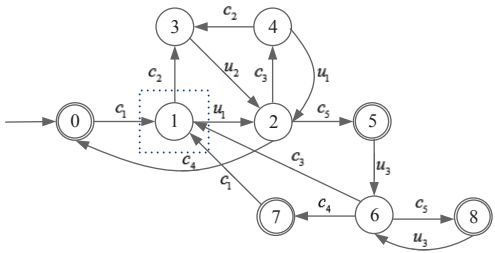}
\caption{State transition graph of the DBA representing supremal sublanguage $\sup\mathcal{M}(E_l)$}
\label{fig:robot_mel}
\end{figure}

\section{$\omega-$nonblocking supervisor synthesis}
For the NSCP$^\omega$, we have the following result:
\begin{thm} \label{pro:existcond}
For any plant DES $\bf G$, any $*-$language $K\subseteq L_m({\bf G})$ and $\omega-$language
$T\subseteq S({\bf G})$, there exists an $\omega-$nonblocking supervisor $f:L({\bf G})\rightarrow \Gamma$
that synthesizes $K$ and $T$ if and only if
\begin{enumerate}[(i)]
\item $K$ is $*-$controllable with respect to ${\bf G}$ and $*-$closed with respect to $L_m({\bf G})$;
\item $T$ is $\omega-$controllable with respect to ${\bf G}$ and $\omega-$closed with respect to $S({\bf G})$;
\item $T$ is markable with respect to $L_m({\bf G})$;
\item $\overline{K} = \text{pre}(T)$.

\end{enumerate}
\end{thm}

\begin{proof}
We have known in \cite{Wonham16a} that there exists a nonblocking supervisor $f^*:L({\bf G})\rightarrow \Gamma$
that synthesizes $K$ if and only if condition (i) is satisfied, and in \cite{Thistl91} that
there exists a complete and deadlock-free supervisor $f^\omega:L({\bf G})\rightarrow \Gamma$
if and only if condition (ii) is satisfied.

From the definition of markability, string $s' \in \text{pre}(s)$ visits a marker state in $\bf G$
and this marker state will be visited by $s$ infinitely often; thus ${\bf G}^f$ is livelock-free
if and only if $T = S({\bf G}^f)$ is markable (condition (iii)).

If condition (iv) is satisfied, we have $\overline{K} = \text{pre}(T) = L({\bf G}^{f^*})=L({\bf G}^{f^\omega})$.
Hence, if writing $f = f^*=f^\omega$, then $f$ is nonblocking, deadlock-free and livelock-free, i.e. $\omega-$nonblocking.

On the other side, if there exists an $\omega-$nonblocking supervisor $f:L({\bf G})\rightarrow \Gamma$
that synthesizes $K$ and $T$, then we have $\overline{K} = L({\bf G}^{f^*}) = \text{pre}(T)$, i.e. condition (iv) holds.
\end{proof}

Theorem \ref{pro:existcond} gives the sufficient and necessary conditions for the existence of an $\omega-$nonblocking supervisor. 
In Sections \ref{preliminaries} and \ref{markability}, we have shown that there exists the supremal $*-$closed and $*-$controllable sublanguage $\sup\mathcal{C}^*(E_s)$ of a given language $E_s$, the supremal $\omega-$controllable sublanguage $\sup\mathcal{C}^\omega(E_l)$ of $E_l$, and the supremal markable sublanguage $\sup\mathcal{M}(E_l)$ of $E_l$; however, there does not exist the supremal $\omega-$closed sublanguage of $E_l$ as $\omega-$closure is not preserved under arbitrary unions. Hence, there may not exist a supremal sublanguage that satisfies conditions (i)-(iv) in Theorem \ref{pro:existcond} simultaneously.

To tackle the issue that the $\omega-$closure part of condition (ii) in Theorem \ref{pro:existcond} may not hold, in the following process of supervisor synthesis, we first ignore the requirement that the resulting infinite sublanguage be $\omega-$closed with respect to $S({\bf G})$. Then we apply Algorithm 2 in \cite{Wei19} to obtain a maximal subautomaton representing an $\omega-$closed sublanguage. To this end, NSCP$^\omega$ can be solved by the following approach similar to the Thistle's rules.

First, we compute the supremal language pair $(K, T) \subseteq 2^{\Sigma^*\times\Sigma^\omega}$
satisfying conditions (i), (iii) and (iv) in Theorem \ref{pro:existcond} and $T$ being $\omega-$controllable.
\begin{algm}\label{algm:supcon} {\bf Input}: DES ${\bf G}$ with $L({\bf G})$, $L_m({\bf G})$ and $S({\bf G})$, $*-$language $E_s$, $\omega-$language $E_l$.\\
{\bf Output}: Supremal language pair $(K_N, T_N) \subseteq 2^{\Sigma^* \times \Sigma^\omega}$.
\begin{enumerate}[Step i)]

\item Set $K_0 = L_m({\bf G}) \cap E_s$ and $T_0 = \sup\mathcal{M}(S({\bf G}) \cap E_l)$.

\item For $i \geq 1$, apply the standard algorithm in \cite{Wonham16a} to compute the supremal $*-$controllable
and $*-$closed sublanguage $K_i = \sup\mathcal{C}^*(K_{i-1} \cap \text{pre}(T_{i-1}))$.

\item Apply the algorithm in \cite{Thistl91} to compute the supremal $\omega-$language
$T_i = \sup\mathcal{C}^\omega(\text{lim}(\overline{K_i}) \cap T_{i-1})$.
Check if $\overline {K_i} = \text{pre}({T_i})$: if YES, output $K_i$ and $T_i$ and denote $i$ as $N$; otherwise, advance $i$ to $i+1$ and go to Step ii).

\end{enumerate}
\end{algm}

\begin{lem}\label{teminate}
In Algorithm \ref{algm:supcon}, for some $N$, $K_i=K_N,T_i=T_N$ for all $i\geq N$, where $i,N\in\mathbb{N}$.
\end{lem}
\begin{proof}
As
\begin{align*}
K_i &= \sup\mathcal{C}^*(K_{i-1} \cap \text{pre}(T_{i-1}))\\
& \subseteq K_{i-1} \cap \text{pre}(T_{i-1})\\
& \subseteq K_{i-1},
\end{align*}
$K_i$ is a descending chain \cite{Wonham16a}. Moreover, the automata representing $K_i$ are with finite states. Therefore, by Proposition 2 in [Section 2.8, \cite{Wonham16a}], for some $N$, $K_i=K_N$ for all $i\geq N$.

Similarly, as
\begin{align*}
T_i &= \sup\mathcal{C}^\omega(\text{lim}(\overline{K_i}) \cap T_{i-1})\\
& \subseteq \text{lim}(\overline{K_i}) \cap T_{i-1}\\
& \subseteq T_{i-1},
\end{align*}
$T_i$ is a descending chain. Moreover, the automata representing $T_i$ are with finite states. Therefore, by Proposition 2 in [Section 2.8, \cite{Wonham16a}] and $\overline {K_N} = \text{pre}({T_N})$, for some $N$, $T_i=T_N$ for all $i\geq N$.
\end{proof}
As Lemma \ref{teminate} holds, Algorithm \ref{algm:supcon} will terminate after finite iterations.

\begin{pro}\label{Ti_markable}
With $K_0,K_i,T_0,T_i$ as defined in Algorithm 1, $T_i$ is markable with respect to $L_m(\bf G)$.
\end{pro}
\begin{proof}
The proof is by induction.\\
(Basis step): $T_0$ is markable with respect to $L_m(\bf G)$ by its definition.\\
(Inductive step): Assume that $T_{i-1}$ is markable with respect to $L_m(\bf G)$. As $T_i=\sup\mathcal{C}^\omega(\text{lim}(\overline{K_i}) \cap T_{i-1})\subseteq \text{lim}(\overline{K_i}) \cap T_{i-1}\subseteq T_{i-1}$ and $T_{i-1}$ is markable with respect to $L_m(\bf G)$ by the inductive assumption, by Lemma 1 we have that $T_i$ is markable with respect to $L_m(\bf G)$.
\end{proof}

It is easily verified that the output language pair $(K_N,T_N) \in 2^{\Sigma^*\times\Sigma^\omega}$ of Algorithm 1 satisfies that
\begin{align}
K_N &= L_m({\bf G}) \cap \overline{K_N}, \label{algm:output1} \\
K_N &\subseteq L_m({\bf G})\cap E_s, \label{algm:output2}\\
T_N &\subseteq S({\bf G})\cap E_l, \label{algm:output3}\\
\overline{K_N} &= \text{pre}(T_N). \label{algm:output4}
\end{align}
Let
\begin{align}\label{sup_E_s_E_l}
\nonumber\mathcal{C}^{*\omega}&(E_s,E_l)= \{(K,T)|\\
\nonumber&K\subseteq E_s ~\text{is $*-$controllable with respect to ${\bf G}$ and}\\
\nonumber&\text{$*-$closed with respect to $L_m({\bf G}$)}, \\
\nonumber&\&~T\subseteq  E_l ~\text{is $\omega-$controllable with respect to ${\bf G}$},\\
\nonumber&\&~T\subseteq  E_l ~\text{is markable with respect to $L_m({\bf G})$},\\
&\&~\overline{K} = \text{pre}(T)\}.
\end{align}
\begin{pro}\label{sup_star_omega0}
$\mathcal{C}^{*\omega}(E_s,E_l)$ is nonempty and is closed under arbitrary unions. Moreover, $\mathcal{C}^{*\omega}(E_s,E_l)$ contains a unique supremal element, denoted as $\sup\mathcal{C}^{*\omega}(E_s,E_l)$.
\end{pro}
\begin{proof}
Since the empty language is $*-$controllable, $*-$closed, $\omega-$controllable and markable, it is a member of $\mathcal{C}^{*\omega}(E_s,E_l)$.

Let $(K_\alpha,T_\alpha)\in \mathcal{C}^{*\omega}(E_s,E_l)$ for all $\alpha$ in some index set $A$, and let $(K,T)=\cup\{(K_\alpha,T_\alpha)|\alpha\in A\}$. Then, $(K,T)\subseteq(E_s,E_l)$. Furthermore, $K$ is $*-$controllable with respect to ${\bf G}$ by Proposition 1 in [Section 3.5, \cite{Wonham16a}], $T$ is $\omega-$controllable with respect to ${\bf G}$ by Proposition 5.8 in \cite{Thistl91}, and $T$ is markable with respect to $L_m({\bf G})$ by Lemma \ref{lem:sup_markability} and Proposition \ref{Ti_markable}. Thus, we have for the supremal element
\begin{align*}
\sup\mathcal{C}^{*\omega}(E_s,E_l)=\cup\{(K,T)|(K,T)\in\mathcal{C}^{*\omega}(E_s,E_l)\}.
\end{align*}
\end{proof}
\begin{pro}\label{sup_star_omega}
With $\mathcal{C}^{*\omega}(E_s, E_l)$ as defined in (\ref{sup_E_s_E_l}), $K_0,K_i,T_0,T_i$ as defined in Algorithm 1 and $K_N,T_N$ being its output languages, we have \[(K_N,T_N) = \sup\mathcal{C}^{*\omega}(E_s, E_l).\]
\end{pro}
\begin{proof}
$(\subseteq)$ As $K_N$ is a $*-$controllable and $*-$closed sublanguage of $E_s$, $T_N$ is an $\omega-$controllable sublanguage of $E_l$, $T_N$ is markable, and $\overline {K_N} = \text{pre}({T_N})$, this direction is automatic by the definition of $\sup\mathcal{C}^{*\omega}(E_s,E_l)$.

$(\supseteq)$ Let $(K',T'):=\sup\mathcal{C}^{*\omega}(E_s,E_l)$. We need to prove that $K'\subseteq K_N,T'\subseteq T_N$.\\
First we show that $\overline{K'}\subseteq \overline{K_N}$ by contradiction.\\
Assume that $t\in\overline{K'},t\notin\overline{K_N}$. As $t\in \overline{K'}\subseteq \overline{K_0}$, there exists some $j\in\{1,\cdots,N\}$ such that $t\in \overline{K_{j-1}}$ but $t\notin \overline{K_j}= \overline{\sup\mathcal{C}^*(K_{j-1} \cap \text{pre}(T_{j-1}))}$, which includes two cases:
\begin{enumerate}[(i)]
  \item $t\notin\text{pre}(T_{j-1})$. We have $t\in \text{pre}(T_0)$ as $t\in \overline{K'}=\text{pre}(T')\subseteq\text{pre}(T_0)$. With $t\in \overline{K_{j-1}}$ we have $t\in\text{pre}(T_{j-2})$ by the definition of $K_{j-1}$. The reason for $t\in\text{pre}(T_{j-2})$ but $t\notin\text{pre}(T_{j-1})=\text{pre}(\text{sup}\mathcal{C}^\omega(\text{lim}(\overline{K_{j-1}})\cap T_{j-2}))$ is that there exists string $t'\sigma\leq t,t'\in\Sigma^*,\sigma\in\Sigma$ such that the occurrence of event $\sigma$ will violate the $\omega-$controllability. However, with $t\in\text{pre}(T')$ the occurrence of event $\sigma$ will not violate the $\omega-$controllability as $T'$ is an $\omega-$controllable sublanguage, which is a contradiction.
  \item $t\in \overline{K_{j-1}\cap\text{pre}(T_{j-1})}$ but $t\notin \overline{K_j}$. Then there exists string $t'\sigma\leq t,t'\in\Sigma^*,\sigma\in\Sigma$ such that the occurrence of event $\sigma$ will violate the $*-$controllability. However, with $t\in \overline{K'}$ we have that the occurrence of event $\sigma$ will not violate the $*-$controllability as  $K'$ is a $*-$controllable sublanguage, which is a contradiction.
\end{enumerate}
As the above two cases do not hold, we have $t\in \overline{K_N}$. Thus, $\overline{K'}\subseteq \overline{K_N}$.\\
As $K',K_N$ are $*-$closed with respect to $L_m({\bf G})$, we have $K'=\overline{K'}\cap L_m({\bf G}),K_N=\overline{K_N}\cap L_m({\bf G})$. As $\overline{K'}\subseteq \overline{K_N}$, $K'\subseteq K_N$ follows directly.\\
By $\overline{K'}=\text{pre}(T')$, we have $T'\subseteq\text{pre}^{-1}(\overline{K'})\cap\Sigma^\omega=\text{lim}(\overline{K'})$. As $\overline{K'}\subseteq\overline{K_N}$ and $K_i,i\in\{0,\cdots,N\}$ is a descending chain, we have $T'\subseteq\text{lim}(\overline{K'})\subseteq \text{lim}(\overline{K_i})$.\\
Then we show that $T'\subseteq T_N$ by contradiction.\\
Assume that $s\in T',s\notin T_N$. We have $s\in T_0=\sup\mathcal{M}(S({\bf G})\cap E_l)$; otherwise $s\notin T'$ as $T'$ is a markable sublanguage. Then there exists some $j\in\{1,\cdots,N\}$ such that $s\in T_{j-1}$ but $s\notin T_j$. As $T_j=\sup\mathcal{C}^\omega(\text{lim}(\overline{K_j})\cap T_{j-1})$, $s\in \text{lim}(\overline{K_j})\cap T_{j-1}$ but $s\notin T_j$ implies that there exists string $t\sigma\leq s,t\in\Sigma^*,\sigma\in\Sigma$ such that the occurrence of event $\sigma$ will violate the $\omega-$controllability. However, with $t\sigma\leq s\in T'$ and $T'$ is an $\omega-$controllable sublanguage, the occurrence of event $\sigma$ will not violate the $\omega-$controllability, which is a contradiction. Thus, we have $T'\subseteq T_N$.
\end{proof}

With $(K_N, T_N)$, we have the following result.

\begin{thm}  \label{thm:mark}
$NSCP^\omega$ is solvable if and only if the language pair $K_N$ and $T_N$ returned
by Algorithm~\ref{algm:supcon} are both nonempty, and \[\inf\mathcal{F}^\omega(A_l) \subseteq T_N.\]
\end{thm}

\begin{proof}
First, by \cite[Theorem 5.9]{Thistl91}, since $T_N \neq \emptyset$ and $\inf\mathcal{F}^\omega(A_l) \subseteq T_N$,
there exists a complete and deadlock-free supervisor $f:L({\bf G})\rightarrow \Gamma$ (the detailed
rules are referred to \cite{Thistl91}) that solves SCP$^\omega$. Namely, the behaviors of the controlled
plant ${\bf G}^f$  are:
\begin{align*}
(a) &L({\bf G}^f) = \text{pre}(A') \cup \bigcup_{m\in M}m\text{pre}(E_m'),\\
(b) &S({\bf G}^f) = A' \cup \bigcup_{m\in M}mE_m',
\end{align*}
where $A' = \inf\mathcal{F}^\omega(A_l)$, $M$ is the set of all elements of $\text{pre}(T_N)\setminus \text{pre}(A')$ of minimal length
and $E_m'$ is a sublanguage of $T_N/m$ synthesized by a branch of the supervisor $f$.
Also, the behaviors satisfy the conditions
\begin{align}
&A' \subseteq S({\bf G}^f) \subseteq T_N, \label{eq:inf_cond1}\\
&\text{pre}(S({\bf G}^f)) = L({\bf G}^f). \label{eq:inf_cond2}
\end{align}

As $S({\bf G}^f) \subseteq T_N$ and $T_N$ is markable with respect to $L_m({\bf G})$, by Lemma \ref{lem:mark_sub} we have that $S({\bf G}^f)$ is markable with respect to $L_m({\bf G})$, i.e. ${\bf G}^f$ is livelock-free.

Under the control of $f$, the finite marker behavior is $L_m({\bf G}^f) = L_m({\bf G})\cap L({\bf G}^f)$.
It is left to prove that (i) $\overline{L_m({\bf G}^f)} = L({\bf G}^f)$ and (ii) $L_m({\bf G}^f) \subseteq E_s$.
For (i),
\begin{align*}
\overline{L_m({\bf G}^f)} & = \overline{L_m({\bf G})\cap L({\bf G}^f)} \\
                          & = \overline{L_m({\bf G})\cap \text{pre}(S({\bf G}^f))} ~~(\text{by (\ref{eq:inf_cond2})})\\
                          & = \text{pre}(S({\bf G}^f)) ~~(\text{by Lemmas \ref{lem:mark_sub} and \ref{lem:mark_nonblock}}) \\
                          & = L({\bf G}^f).
\end{align*}
For (ii),
\begin{align*}
L_m({\bf G}^f) & = L_m({\bf G})\cap L({\bf G}^f)\\
               & = L_m({\bf G}) \cap \text{pre}(S({\bf G}^f)) ~~(\text{by (\ref{eq:inf_cond2})})\\
               & \subseteq L_m({\bf G})\cap \text{pre}(T_N)  ~~(\text{by (\ref{eq:inf_cond1})})\\
               & = L_m({\bf G})\cap \overline{K_N} ~~(\text{by (\ref{algm:output4})})\\
               & = K_N ~~(\text{by (\ref{algm:output1})}) \\
               & \subseteq E_s.
\end{align*}
In conclusion, the constructed supervisor $f:L({\bf G})\rightarrow \Gamma$ is complete, $\omega-$nonblocking (nonblocking, deadlock-free and livelock-free), and the controlled behaviors satisfy conditions (\ref{eq:NSCP_cond1}) and (\ref{eq:NSCP_cond2}),
i.e. NSCP$^\omega$ is solvable.
\end{proof}

In the example of the robot, Algorithm 1 terminates when $N=2$. The resultant supervisor is shown in Fig.~\ref{fig:robot_sup}. Event $c_2$ is disabled at states $1$ and $3$ to meet the safety specification. For the maximal legal specification, state $2$ is visited infinitely often, which ensures that state $1$ in Fig.~\ref{fig:robot_el} is visited infinitely often. Moreover, the supervisor is livelock-free as any infinite string will visit some marker states.

\begin{figure}[ht]
\centering
    \includegraphics[scale = 0.5]{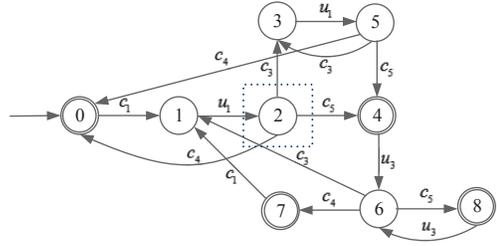}
\caption{State transition graph of the supervisor computed by Algorithm 1}
\label{fig:robot_sup}
\end{figure}

Let DBA ${\bf T}_N$ represent the $\omega-$language $T_N$ returned by Algorithm 1. As $\omega-$closure is not preserved under arbitrary unions, we apply Algorithm 2\footnotemark[1] in \cite{Wei19} to ${\bf T}_N$ to obtain its maximal subautomaton representing an $\omega-$closed sublanguage, which is an $\omega-$nonblocking and $\omega-$closed supervisor.
\footnotetext[1]{The input of Algorithm 2 in \cite{Wei19} is an arbitrary DBA; the output of it is a DBA with all \emph{bad cycles} (i.e. cycles containing no states in the B\"uchi acceptance criterion) being deleted. The intuition of this algorithm is to first detect all loops, and then delete all bad cycles by deleting their \emph{nearest controllable back edges} and removing the uncoreachable transitions iteratively.}
In the example of the robot, we apply Algorithm 2 in \cite{Wei19} to the supervisor shown in Fig.~\ref{fig:robot_sup}. Then we obtain the final supervisor shown in Fig.~\ref{fig:robot_sup_final}, in which event $c_3$ is disabled at state $5$ and event $c_5$ is disabled at state $6$ to prevent bad cycles $(3,5)^\circlearrowleft$ and $(6,8)^\circlearrowleft$ respectively.
\begin{figure}[ht]
\centering
    \includegraphics[scale = 0.5]{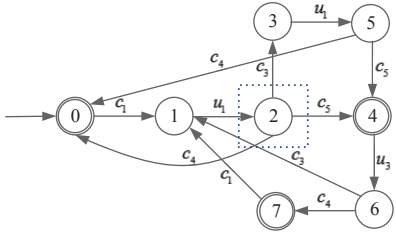}
\caption{State transition graph of the final supervisor}
\label{fig:robot_sup_final}
\end{figure}

\begin{remark}\label{rem_seow}
We note that in \cite{Seow20} a marker-progressive supervisory control approach (the Seow's approach) is proposed with the linear-time temporal logic control of a class of fair DES, in which every marker condition is true infinitely often. The Seow's approach sets some uncontrollable events as fair events and the marker progress under supervised temporal safety is achieved by the DES event fairness. Namely, the event-selection mechanism of the supervisor is directed by a fair event subset to drive the DES to visit every marker condition infinitely often.

A schematic about problem formulations and solutions of the Thistle's approach, the Seow's approach and ours is illustrated in Fig.~\ref{fig:relations}. In all, the Seow's approach and the Thistle's approach are in different language settings, and our approach follows the Thistle and Wonham's supervisory control framework.

In the Thistle's approach, the safety and liveness specifications are described by $\omega-$languages (often represented by Rabin automata). The infinite behavior of the supervisor computed by the Thistle's approach could not guarantee that the marker states will be visited infinite often.

\begin{figure*}[ht]
\centering
    \includegraphics[scale = 0.8]{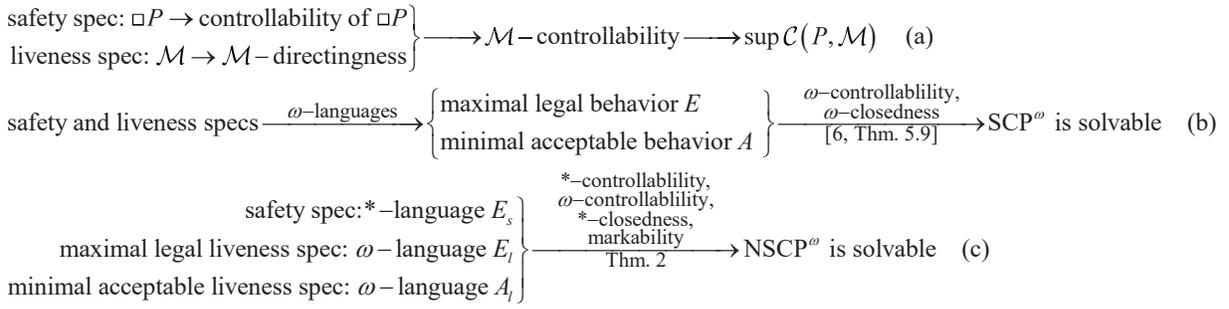}
\caption{Schematic about problem formulations and solutions of (a) the Seow's approach, (b) the Thistle's approach and (c) our approach}
\label{fig:relations}
\end{figure*}

In the Seow's approach, the safety specification and liveness specification are depicted by the invariance of past formula $\Box P$ and the system marker set $\mathcal{M}$ respectively. Moreover, $\mathcal{C}(P,\mathcal{M})$ contains a unique supremal element if $\Box P$ is $\mathcal{M}-$controllable; namely, $\Box P$ is controllable and $\mathcal{M}-$directing. Analogously, the definitions of controllability and $\mathcal{M}-$directingness of $\Box P$ in the Seow's approach are similar to the definitions of $\omega-$controllability and $\omega-$closedness in the Thistle's approach respectively. Under the condition of $\mathcal{M}-$directingness, if we treat the set of marker states as the B\"uchi acceptance criterion, the legal state trajectories of DES model $G$ satisfying $\Box P$ are $\omega-$closed. Thus, in this scenario, $\mathcal{M}-$directingness is a stronger condition than $\omega-$closedness. In the example of the robot, for the existence of cycle $(2,3)^\circlearrowleft$, it is not $\omega-$closed. Even without this cycle, it is still not $\mathcal{M}-$directing as cycle $(1,2)^\circlearrowleft$ containing no marker states.

To implement the Seow's approach, the set of fair events has to be selected according to the system marker set. Moreover, to satisfy $\mathcal{M}-$directingness, the plant is restricted to be $\omega-$closed.

We observe that the Seow's approach is not applicable to the example of the robot for two reasons.
\begin{enumerate}
  \item The set of fair events $\Sigma_\mathcal{F}=\Sigma_\mathcal{C}\cup\Sigma_\mathcal{J}$ with the specification pair $(P,\mathcal{M})$ needs to be defined for the event-selection mechanism, where $P\equiv\overline{p_{3}},\mathcal{M}=\{p_0,p_4\}$. However, in this example, events $c_4$ and $c_5$ are not potential fair events as they are not uncontrollable.
  \item As legal state trajectories (satisfying $\Box P$), say $c_1(u_1c_3)^\omega$, fail to meet marker conditions in $\mathcal{M}$ infinitely often, $\Box P$ is not $\mathcal{M}-$directing.
\end{enumerate}

To apply the Seow's approach, we may modify controllable events $c_4$ and $c_5$ to uncontrollable events (say $u_4$ and $u_5$ respectively). By setting $\Sigma_\mathcal{C}=\{u_4,u_5\}$ and employing the event-selection mechanism based on the fair event subset, we obtain the resultant supervisor shown in Fig.~\ref{fig:seow_sup}.
\end{remark}

\begin{figure}[ht]
\centering
    \includegraphics[scale = 0.5]{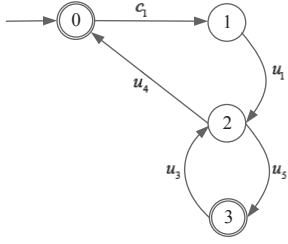}
\caption{State transition graph of the the supervisor returned by the Seow's approach}
\label{fig:seow_sup}
\end{figure}

While in our approach, we employ $*-$language and $\omega-$language to represent safety specification and liveness specification respectively. Generally, with Theorem 2 we have a sufficient and necessary condition to compute a supremal $\omega-$controllable and markable sublanguage, while Proposition 7 in \cite{Seow20} gives a sufficient condition for the existence of a supremal element. In addition, to obtain a supervisor guaranteeing that infinite strings of the controlled DES with infinite behavior will always some marker states infinitely often, by our approach no additional restrictions on events are necessary and the plant is not required to be $\omega-$closed. As both two approaches tackle the properties relevant to marker states, the result of the Seow's approach is within the scope of ours.

\begin{remark}
We also note that a compositional supervisory control approach is proposed in \cite{Ciolek20}, in which the supervisory control problem is translated into the reactive synthesis and planning frameworks. Applying the reactive synthesis approach to the example of the robot, we have a path of a CTL formula $\varphi_\mathcal{M}$ generated in the translation depicted in \cite{Ciolek20} as $0\xrightarrow{\textsf{c}=c_1}1\xrightarrow{\textsf{u}=u_1}2\xrightarrow{\textsf{c}=c_4}0$, where symbol $\textsf{c}$ encodes the system's choice of an event and symbol $\textsf{u}$ encodes the environment's choice of an uncontrollable event or no choice. This path guarantees that a marker state is reached from initial state $0$ (i.e. nonblocking). As the supervisor for $\mathcal{M}$ (a director) enables at most one controllable event at every state, the result in \cite{Ciolek20} is usually not supremal. In contrast, Algorithm 1 in our approach returns a more permissive $\omega-$controllable sublanguage by following the Thistle and Wonham's supervisory control framework.
\end{remark}

\begin{remark}
Reference \cite{Schmuck20} presents how to solve a synthesis problem from supervisory control theory by reactive synthesis or vice-versa. Here the synthesis problem is the same as the problem formulation of the supervisory control theory proposed by Ramadge \cite{Wonham16a} and Thistle and Wonham \cite{ThiWon94b}. In the example of the robot, by Theorem 2 in \cite{Schmuck20}, the reactive synthesis algorithm returns an equivalent solution with the Thistle's approach.
\end{remark}

\section{Conclusions}
In this paper, we propose an approach to synthesize an $\omega-$nonblocking supervisor in the framework initiated by Thistle and Wonham. The resultant supervisor ensures that any infinite string of the controlled DES with infinite behavior will visit some marker states. In the future work, we will extend our method to the decentralized and distributed supervisory control of DES with infinite behavior.

\end{document}